\documentclass[a4paper,10pt]{article}

\usepackage{fullpage}
\usepackage{authblk}
\usepackage[numbers, sort&compress]{natbib}

\usepackage{xcolor}
\usepackage{algorithm} 
\usepackage{algpseudocode}
\usepackage{microtype}
\usepackage{graphicx} 
\usepackage{researchpack} 
\usepackage{url} 
\usepackage{hyperref}
\usepackage{dsfont} 
\usepackage[capitalize,nameinlink]{cleveref} 
\usepackage{comment}
\usepackage{caption}
\usepackage{subcaption}
\usepackage{multirow}
\usepackage[symbol]{footmisc}  

\newcommand{\Yes}{\raisebox{-0.2em}{\includegraphics[width=2.5ex]{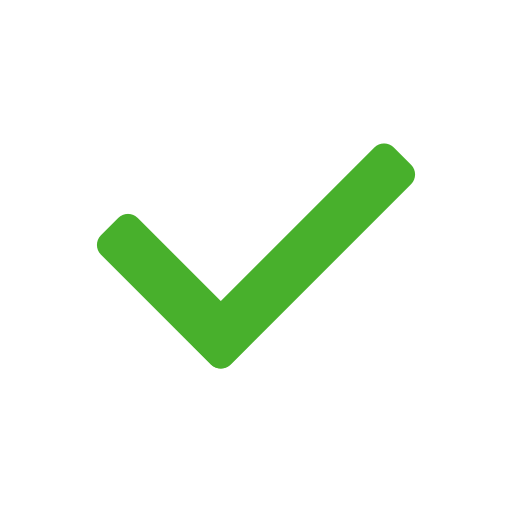}}\xspace}
\newcommand{\No}{\raisebox{-0.2em}{\includegraphics[width=2.5ex]{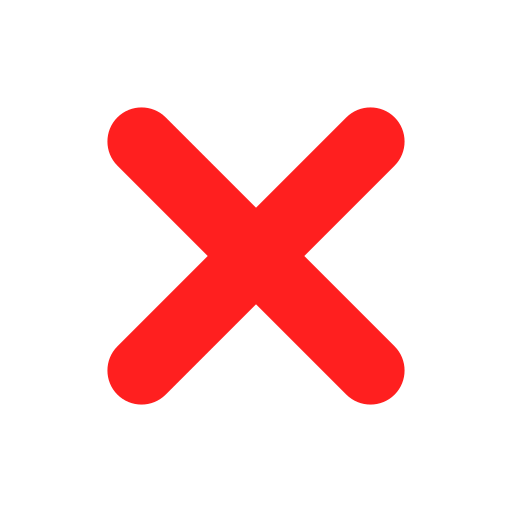}}\xspace}

\newtheorem{property}{Property}

\definecolor{green_colorblind}{HTML}{029E73}
\definecolor{orange_colorblind}{HTML}{D55E00}
\definecolor{darkred}{rgb}{0.55, 0.0, 0.0}

\newcommand{\Tsafe}{T^{(\text{safe})}}
\newcommand{\Treplace}{T^{(\text{replace})}}

\title{You Don't Bring Me Flowers: Mitigating Unwanted Recommendations Through Conformal Risk Control\thanks{Accepted at the 19th ACM Conference on Recommender Systems (RecSys 2025).}}

\author[1]{Giovanni De Toni} 
\author[2]{Erasmo Purificato} 
\author[2]{Emilia Gomez} 
\author[1]{Bruno Lepri} 
\author[3]{Andrea Passerini} 
\author[2]{Cristian Consonni} 

\affil[1]{Fondazione Bruno Kessler, \texttt{\{gdetoni,lepri\}@fbk.eu}}
\affil[2]{European Commission, Joint Research Centre (JRC)\thanks{\textbf{Disclaimer}: The view expressed in this paper is purely that of the authors and may not, under any circumstances, be regarded as an official position of the European Commission.} \texttt{\{erasmo.purificato,emilia.gomez-gutierrez,cristian.consonni\}@ec.europa.eu}}
\affil[3]{DISI, University of Trento, \texttt{andrea.passerini@unitn.it}}

\date{\vspace{-10mm}}

\begin{document}

\maketitle

\begin{abstract}
Recommenders are significantly shaping online information consumption.
While effective at personalizing content, these systems increasingly face criticism for propagating irrelevant, unwanted, and even harmful recommendations.
Such content degrades user satisfaction and contributes to significant societal issues, including misinformation, radicalization, and erosion of user trust.
Although platforms offer mechanisms to mitigate exposure to undesired content, these mechanisms are often insufficiently effective and slow to adapt to users' feedback.
This paper introduces an intuitive, model-agnostic, and distribution-free method that uses conformal risk control to provably bound unwanted content in personalized recommendations by leveraging simple binary feedback on items.
We also address a limitation of traditional conformal risk control approaches, \ie the fact that the recommender can provide a smaller set of recommended items, by leveraging implicit feedback on consumed items to expand the recommendation set while ensuring robust risk mitigation.
Our experimental evaluation on data coming from a popular online video-sharing platform demonstrates that our approach ensures an effective and controllable reduction of unwanted recommendations with minimal effort.
The source code is available here: \url{https://github.com/geektoni/mitigating-harm-recsys}.

\end{abstract}

\section{Introduction}

The widespread use of recommender systems on online platforms has significantly changed how users engage with content.
These systems are integral to modern digital media, facilitating personalized navigation through extensive online content across domains such as e-commerce~\cite{ghanem2022balancing,loukili2023machine}, social media~\cite{gao2023surveygnns,sanchez2021social}, and content streaming~\cite{deldjoo2021multimedia,elahi2022towards}, and tailoring experiences based on historical behaviors and expressed preferences~\cite{fernandez2020recsysmisinformation,boratto2021connecting}.
By leveraging user-generated data, recommenders aim to enhance user engagement~\cite{xue2023prefrec}, satisfaction~\cite{nguyen2018user}, and platform profitability~\cite{cai2019trustworthy}.
However, despite their proven utility and benefits, the basic concepts that make recommendation algorithms effective, such as \textit{similarity-based filtering} and \textit{preference prediction}, can inadvertently contribute to the spread of unwanted content~\cite{liu2024train}.
Platforms have faced increasing criticism for creating \textit{filter bubbles} and \textit{rabbit holes} that reinforce existing beliefs, and for spreading potentially dangerous or harmful content~\cite{tomlein2021auditmisinformation}.
Such dynamics are often attributed to algorithms that prioritize \textit{engagement}, sometimes at the expense of content quality and truthfulness~\cite{tommasel2021ohars}.
Features like the ``\textit{Not Interested}'' button on YouTube~\cite{liu2024train} suffer from issues like low user awareness, complexity of use, or limited effectiveness.
The lack of transparency and explainability in such strategies further exacerbates these problems, making it difficult for users to understand why they are being shown certain content and how to avoid it.
As a result, users develop ``\textit{folk theories}'' to explain the recommendations they are getting and try to influence the recommender outcomes~\cite{bernstein2013quantifying, eslami2015always, eslami2016first,liu2024train}.
Recent works have studied how to mitigate unwanted content, biases, or ensure fairness in recommendations \cite{chee2024harm, ahn2024interactive, ali2021harmadvertising, morikControllingFairnessBias2020}, but it remains unclear how to grant users precise and transparent control over the systems' \textit{risk}. 

In this paper, we propose a post-hoc method based on conformal prediction~\cite{shafer2008tutorial,angelopoulos2023gentleintro} to provably mitigate the exposure to unwanted recommendations \textit{in expectation}.
Conformal prediction is a framework for uncertainty quantification for machine learning models that constructs \textit{prediction sets} containing the true output with a \textit{user-specified confidence level}, assuming data exchangeability.
It uses \textit{non-conformity scores} from model outputs to provide \textit{distribution-free}, \textit{finite-sample} guarantees.
{\color{black}
Conformal methods have recently been applied to decision support in classification tasks~\cite{straitouri2023improving, straitouri2024counterfactualpset, detoni2024predictionset}, natural language processing~\cite{campos2024surveyconformalnlp}, and recommendation systems~\cite{angelopoulos2023recommendation, xu2024conformal}.
However, in recommender systems, current approaches often overlook key challenges, such as re-recommending unwanted items or balancing personalization with content filtering.
}
To address this, we apply \textit{conformal risk control}~\cite{angelopoulos2022conformal}, which generalizes conformal prediction to control \textit{general loss functions} (or “\textit{risk}”) rather than coverage.
Our method builds recommendation sets by \textit{filtering} potentially unsafe items and \textit{replacing} them with previously consumed, \textit{safe} content, while provably respecting a user-defined risk level.
Specifically, our contributions include: 
\begin{enumerate}
    \item A comprehensive analysis of real data to understand \textit{reporting patterns} and watch time behavior related to unwanted content, providing insights into how users interact with and respond to such content.
    \item A method based on conformal risk control that provably limits unwanted content in recommendations, with theoretical guarantees, by leveraging \textit{repeated content}.
    \item A practical heuristic to choose \textit{safe} alternative content, balancing the need for personalization and the goal of minimizing unwanted recommendations.
    \item Experiments and ablations with a real-world dataset, highlighting our approach's effectiveness in significantly reducing unwanted content.
\end{enumerate}

\section{Related Work}

Recommender systems have faced growing criticism for promoting unwanted content such as misinformation~\cite{fernandez2020recsysmisinformation}, hate speech~\cite{badjatiya2017deep}, and offensive or irrelevant items~\cite{ali2021harmadvertising,yi2024isolationmultimodalrecs}.
These issues often stem from algorithmic bias~\cite{baeza2020bias,boratto2021advances}, en\-ga\-ge\-ment-driven objectives~\cite{hauptmann2022research}, and feedback loops that reinforce filter bubbles~\cite{cinus2022effect,tommasel2021want}.
Further, repeated exposure to undesirable recommendations impacts user trust and satisfaction~\cite{ytre-arne2021folk}, contributing to what researchers describe as ``algorithmic hate''~\cite{smith2022recsyshate}.
Unfortunately, efforts to empower users, such as feedback tools, are often hindered by low visibility and unclear system responses~\cite{liu2024train,bernstein2013quantifying}.

To address these challenges, prior work has proposed content moderation~\cite{schneider2023effectiveness}, fairness constraints~\cite{zehlike2022fairness}, and diversity-promoting techniques~\cite{tommasel2021ohars,smith2022recsyshate}.
While effective to a degree, these methods often lack formal guarantees on limiting exposure to harmful content and seldom incorporate fine-grained user behavior~\cite{kasirzadeh2023user,ge2024survey}.
Our approach is closely related to~\cite{morikControllingFairnessBias2020,chee2024harm} but differs in key ways: we provide item-level control over unwanted content exposure -- beyond group-level guarantees~\cite{morikControllingFairnessBias2020} -- and avoid the need to learn custom policies from user feedback ~\cite{chee2024harm}.
Recent work has introduced conformal methods to recommender systems~\cite{angelopoulos2023recommendation,xu2024conformal}, focusing on controlling metrics like Normalized Discounted Cumulative Gain (nDCG) or False Discovery Rate (FDR). However, to the best of our knowledge, no prior method has used conformal techniques to control the fraction of unwanted content in recommendations.

Lastly, repeated content consumption is a well-established behavior across platforms~\cite{anderson2014dynamics, ren2019repeat}.
Recent approaches have leveraged this phenomenon through repeat-explore strategies~\cite{ren2019repeat}, neural ODEs for temporal dynamics~\cite{dai2024recode}, and collaborative filtering models that capture item-specific repeat patterns with varying lifetimes~\cite{wang2019modeling}.
We exploit this phenomenon to offer formal guarantees on risk control while preserving performance.

\begin{figure*}
     \centering
     \begin{subfigure}[t]{0.7\textwidth}
          \centering
        \includegraphics[width=\linewidth]{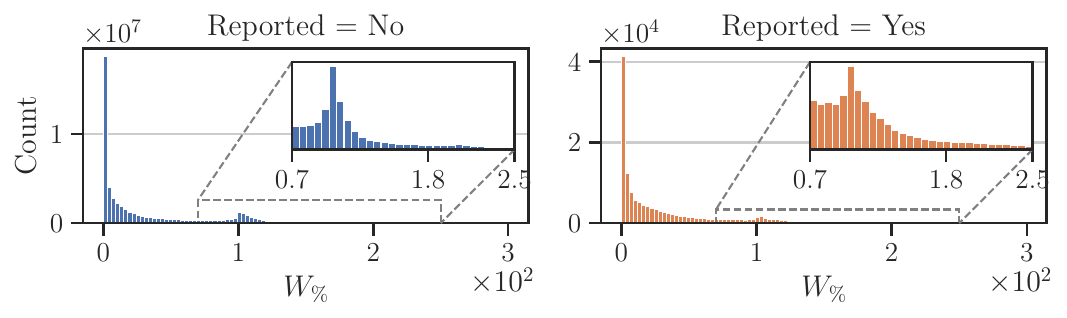}
        \caption{Distribution of $W_{\%}$ on reported (right) and unreported (left) videos.}
        \label{fig:user-fraction-time}
     \end{subfigure}
     \begin{subfigure}[t]{0.7\textwidth}
         \centering
          \includegraphics[width=\linewidth]{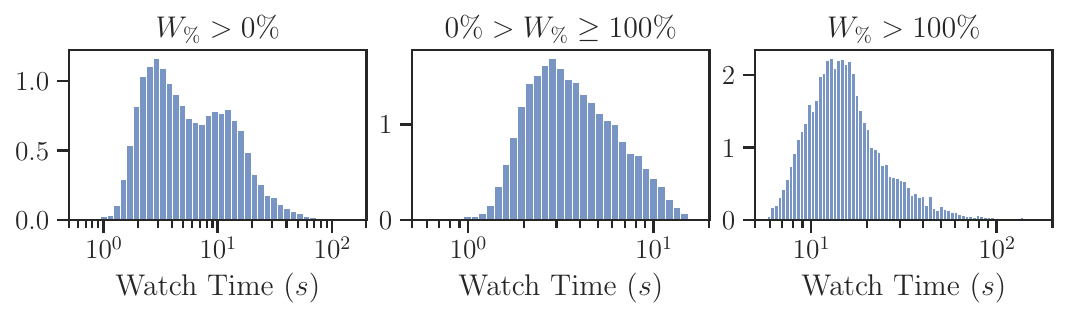}
          \caption{Reported videos watch time segmented by $W_{\%}$ (log-scale).}
          \label{fig:dwell_distribution_reported}
     \end{subfigure}
     \caption{Statistics on the user behaviour on Kuaishou data.}
\end{figure*}

\section{Case Study: A Video-Sharing Platform}
\label{sec:case-study-kuaishou}

We begin our study by examining data from a popular Chinese short-form video-sharing platform, \textit{Kuaishou}\footnote{\url{https://en.wikipedia.org/wiki/Kuaishou}}, which has over 300 million daily users at the time of writing.
We use the \textit{KuaiRand} dataset~\cite{kuairand2022}, which was collected by sampling 27000 Kuaishou users who collectively watched approximately 32 million videos over one month (from April 8 to May 8, 2022). This dataset contains various user interaction signals, such as watch time for each video and whether users liked or commented on a video.
Unfortunately, the dataset does not include Kuaishou’s ranking scores or the explicit recommendation order of videos.
However, the dataset also includes \textit{negative feedback} signals, such as when users click the \textit{“Do not recommend”} or \textit{“Report”} buttons while watching a video. Similar negative feedback data is available in the \textit{YouTube Regrets} dataset\footnote{\url{https://huggingface.co/datasets/mozilla-foundation/youtube_regrets}}, but only in an aggregated format without user profiles, limiting the scope of analysis.
We focus exclusively on short videos, removing advertisements and videos with zero duration (\eg dynamic photos). Notably, approximately $29\%$ of users provided negative feedback on at least one video during the data collection period. Consequently, our analysis is restricted to this subset.
After further filtering, our dataset consists of 7601 users, 56 million interactions, and 10 million unique videos. On average, a user watched $7,417$ videos during the collection period (median: $4,623$), with a large standard deviation of $9,223$.

\begin{table}[b]
\centering
\caption{Statistics on reported videos.}%
\resizebox{0.5\linewidth}{!}{%
\begin{tabular}{@{}cccc@{}}
\toprule
$H_X$ (\%) & N. of users & $\%$ of users & eCDF \\ \midrule
$0 \leq H < 0.1 $ & 4569 & 60.1 & 0.601 \\
$0.1 \leq H < 0.3 $ & 1909 & 25.1 & 0.852 \\
$0.3 \leq H < 1 $ & 827 & 10.9 & 0.961 \\
$1 \leq H < 100 $ & 296 & 3.9 & 1.0 \\ \bottomrule
\end{tabular}%
}
\label{tab:statistics-reported-videos}
\end{table}

\subsection{Sparsity of Negative Feedback}

We start by considering the propensity of users to give negative feedback to a recommended video seen within the interface.
Thus, we consider the percentage of videos a user reported overall videos they have seen
$
    H_X = \frac{\text{\# of videos reported by user $U$}}{\text{\# of videos seen by user $U$}} \cdot 100
$.
\cref{tab:statistics-reported-videos} shows that $60\%$ of users report at most 1 every 1000 videos they have seen. 
Overall, the negative feedback is very sparse since $> 95\%$ of users only report at most 1 every 100 videos recommended to them.
The relatively low reporting rate may be attributed to the accessibility of the \textit{``Do not recommend''} and \textit{``Report''} options, which require multiple interactions within the interface.
Indeed, many users might simply skip a video they do not like.
Nevertheless, the data sparsity represents a challenge for recommender systems that want to exploit this information to minimize recommendations of unwanted content. 

\subsection{Analysis of the Watch Time}
\label{sec:analysis_dwell_time}

Next, we analyze user watch time behavior to investigate potential differences in the viewing patterns of reported videos. 
We define with $W \in \bbR_{\geq 0}$ the \textit{duration of a user-video interaction} (\eg how much time the user $U$ spent on that video $I$).
Since video durations vary, we measure each user-video interaction as the proportion of the video watched by the user:
\begin{equation}
W_{\%}(U,I) = \frac{W(U, I) \coloneqq \text{View time by user U of video I}}{\text{Duration of video $I$}} \cdot 100
\end{equation}
Figure~\ref{fig:user-fraction-time} shows the distribution of $W_{\%}$ for both reported and unreported videos.
Both distributions peak near zero, indicating that many videos are either skipped or watched very briefly.
We observe that $21\%$ of interactions result in no watch time.
There is also a peak near one, as highlighted, suggesting a different behavior closer to videos viewed in full.
Additionally, some user-video interactions exceed the video's original duration, indicating users occasionally linger on videos.
Similarly to prior research \cite{yin2013silence}, we focus on reported videos and group them based on their duration.
In our analysis, we consider videos with durations between 7 and 12 seconds, which are the most common.
Since the distribution is zero-inflated (many videos are not watched), we analyze only interactions with positive watch time ($W > 0$).
\cref{fig:dwell_distribution_reported} illustrates the distribution of watch time (log scale).
As seen in the left-most plot, there are peaks around zero and near 10 seconds.
Consequently, we model two scenarios: (a) interactions where the video is watched at most in full ($0\% < W_{\%} \leq 100\%$), and (b) interactions where the video is watched for longer than its duration ($W_{\%} > 100\%$).
\cref{fig:dwell_distribution_reported}'s centre and right-most plots show both distributions have approximately normal shapes with left and right tails, respectively.
Such behavior could be modeled by considering the LogNormal or Weibull distributions, which are commonly used for describing dwell time, \eg on the page views \cite{liu2010understanding, yin2013silence}.

\subsection{Trend of Video Reporting Patterns}
\label{sec:video-reporting-patterns}

We also analyze user behavior when rewatching previously viewed videos. On average, $2.6\%$ of videos watched by users had been seen before. Furthermore, as shown by \cref{tab:statistics-repeated-and-reported-videos}, $8\%$ of users are shown videos they previously reported, and another $8\%$ report previously watched videos that they had not reported initially.
As argued in \cref{sec:analysis_dwell_time}, the latter behavior suggests that users may have skipped the video on the first encounter or found the reporting process cumbersome (\eg requiring multiple clicks).
This hypothesis is supported by \cref{tab:statistics-repeated-and-reported-videos} that shows that the $75\%$ percentile of the watch time for previously seen videos reported anew (2nd row) resembles the $75\%$ percentile of initially reported videos (4th row). 
Interestingly, reported videos that get flagged a second time show a very similar watch time (3rd row). We postulate this is stronger feedback indicating items that the user strongly considers unwanted. 
These findings may be influenced by the platform's continuous stream of new content, which reduces the likelihood of users encountering the same video multiple times.
Nevertheless, $16\%$ of users either see or report again at least one previously seen video.
Lastly, \cref{tab:statistics-repeated-and-reported-videos} shows how videos seen a second time have a greater average watch time.
For videos reported anew the second time (2nd row), the watch time increases significantly, suggesting that users may deliberately review undesirable content more thoroughly before reporting. 

\paragraph{Summary of the analysis.}
We conclude the analysis of Kuaishou data by underlining some key insights that will be useful in designing harm-controlling recommendation systems:
\begin{itemize}
    \item The engagement and perceived harmfulness of videos are not fully correlated. As shown by \cref{fig:user-fraction-time}, reported videos have the same $W_\%$ distribution as non-reported videos. Thus, \textit{maximizing engagement does not imply providing less harmful videos to users.}
    \item In Kuaishou, the system suggests undesired content \textit{despite explicit user feedback} (cf. \cref{tab:statistics-repeated-and-reported-videos}); 
    \item Given that over 95\% of users report less than 1\% of videos they watch, recommender systems should maximize the usefulness of this little explicit negative feedback. 
    \item With 21\% of interactions resulting in no watch time and a strong peak near zero in view time distribution (cf. \cref{sec:analysis_dwell_time}), very short engagement can serve as an early indicator of harmful or undesired content.
\end{itemize}

\subsection{Unwanted, Disliked and Harmful Content}
\label{sec:unwanted-disliked-harmful}

Given the results of \cref{sec:case-study-kuaishou}, we want to emphasize how understanding user feedback is fundamental for the task at hand.
Indeed, online platforms provide various mechanisms for users to express feedback on recommended items.
Recent research shows that mobile app affordances lead to ``\textit{triggered essential reviewing}"~\cite{piccoli2016triggered}, where users provide quick, focused, and often emotionally charged feedback immediately after an experience.
In this context, it is crucial to recognize that users can express feedback through multiple channels, each carrying distinct and sometimes diverging meanings.
For instance, \textit{flagging an item as unwanted} and \textit{reporting it as harmful} are typically separate actions within a service interface.
Moreover, to effectively evaluate mitigation techniques in realistic settings, it is essential to consider feedback signals that distinguish between \textit{unwanted} or \textit{harmful items} and those that simply fail to meet user expectations.
For example, our analysis of the KuaiRand dataset reveals that videos with high watch times (suggesting positive engagement) may still be explicitly flagged as unwanted (\cf \cref{tab:statistics-repeated-and-reported-videos}). 
However, many studies~\cite{chen2024sigformer,seo2022siren,chee2024harm} assume that a rating below a certain threshold (\eg below 3 on a 1–5 scale) implies the item is either unwanted or harmful.
In other datasets, such as the Music Streaming Sessions Dataset (MSSD)~\cite{brost2019music}, item skips indicate unwanted content for a user.
These proxies, as seen, can be insufficient.
To the best of our knowledge, only the KuaiRand dataset \cite{kuairand2022} provides such explicit, disaggregated feedback, despite the widespread availability of these mechanisms on modern platforms.

\begin{table}[t]
\centering
\caption{User's behavior with repeated videos.
Given a video $I$, the user can choose to report it (\No) or not (\Yes).
We show the average watch time (1st time and 2nd time the user sees it), the standard deviation, and the third quartile in brackets.
}
\label{tab:statistics-repeated-and-reported-videos}
\resizebox{0.75\linewidth}{!}{%
\begin{tabular}{lllll}
\toprule
Behaviour & I (\%) & U (\%) & Watch Time 1st (s) & Watch Time 2nd (s)  \\ \midrule
\Yes $\rightarrow$ \Yes & 99.79 & 100.0 & $2.71\pm \scriptstyle 12.15 \;\; (1.39)$ & $10.48\pm \scriptstyle 25.85 \;\; (7.95)$ \\
\Yes $\rightarrow$ \No & 0.11 & 8.69 & $0.50\pm \scriptstyle 7.41 \;\; (0.0)$ & $10.73\pm \scriptstyle 20.85 \;\; (10.13)$ \\
\No $\rightarrow$ \No & 0.07 & 6.33 & $7.74\pm \scriptstyle 15.83 \;\; (6.64)$ & $7.85\pm \scriptstyle 15.86 \;\; (6.72)$ \\ 
\No $\rightarrow$ \Yes & 0.03 & 1.49 & $1.36\pm \scriptstyle 6.91 \;\; (0.0)$ & $6.46\pm \scriptstyle 25.42 \;\; (2.5)$ \\ \bottomrule
\end{tabular}%
}
\end{table}

\section{Recommender Systems With Risk Control}

Given the key insights obtained from the analysis of Kuaishou data, let us now define more formally a \textit{risk-controlling recommender system}. 
Consider a finite set of items $\calI \in \{i_1,\ldots,i_N\}$ and a finite set of users $\calU \in \{u_1,\ldots,u_M\}$, each represented by $d$-dimensional feature vectors, respectively.
Our recommender system outputs a set of items $S(U) \subseteq \calI$ based on some relevance score estimated by a ranker $f(U, I) \in \bbR^{+}$ predicting how much an item $I$ is relevant for a user $U$ based on the information within $U$ and $I$ (\eg in terms of watch-time). 
In two-stage recommender systems, we might perform further \textit{post-processing} activities to \textit{filter} certain items, before presenting the final recommendation to the user.
Similarly to Angelopoulos et al.~\cite{angelopoulos2023recommendation}, let us consider a post-processing step where we employ some threshold $\lambda \in \bbR$ to pick only certain items, based on a score $s: U \times I \rightarrow \bbR$.
For example, consider $s(U, I=i) = f(U, I=i)$, we might want to keep only items with a predicted relevance above a certain base value: 
\begin{equation}
    T_\lambda(U) = \{ i \in \calI : s(U, I=i) \geq \lambda \}
    \label{eqn:set-filtereing-T}
\end{equation}
After post-processing, our objective becomes providing the best set $S_\lambda(U) \subseteq T_\lambda(U)$ maximizing a target metric computed on the recommendations (\eg \textit{serendipity} \cite{ge2010beyond}).
Generally, we are interested in recommending the \textit{most} relevant items to each user, under a budget $k \leq N$ \eg the number of slots in a YouTube main page. 
Let us define with $\pi(\calI) = \{i_1, \ldots, i_N\}$ the order induced by sorting the items given the ranker scores $f(U,I=i_j)$. 
If our ranker is reasonably good, then the optimal solution would be simply picking up to the $k$-th item from $\pi(T_\lambda(U))$:
\begin{equation}
    S_\lambda(U,k) = \{ i_j \in \pi(T_\lambda(U)) : j \leq k \}
\end{equation}
In classical learning-to-rank tasks, we try to learn the optimal ranker $f^* = \argmax_{f \in \calF} \bbE[\ell(S_\lambda(U, k))]$  where $\ell: \calI \rightarrow \bbR$ is our target metric, such as the nDCG.
Besides picking the most relevant items for the user, now we want to \textit{control} the \textit{risk} of the set $\calS_\lambda(X ,k)$ induced by the optimal ranker $f^*$. 
Let us assume we have a risk metric $R: \calI \rightarrow (0,1)$ we would like to provably \textit{bound} below a certain value $\alpha \in (0, 1)$. 
For example, we might want the \textit{fraction} of unwanted content in $S_\lambda(U, k)$ to be below a user-defined value.
Since we want to avoid performing costly operations such as re-training the ranker $f^*$, an intuitive way to reach such an objective is by acting on the \textit{post-processing} step, by finding the optimal $\lambda$ threshold ensuring our \textit{risk} is below a user-defined level $\alpha \in (0,1)$:
\begin{equation}
    \lambda^* = \argmax_{\lambda \in \Lambda} \bbE\left[\ell(S_\lambda(U, k))\right] \qquad \text{s.t.} \quad \bbE[R(S_\lambda(U, k))] \leq \alpha
    \label{eqn:standard-equation}
\end{equation}
In the next sections, we will show how to provably bound the risk in \cref{eqn:standard-equation} for \textit{any} ranker by leveraging the user's negative feedback. 

\subsection{Measuring Unwanted Content}
In our setting, we consider an item \textit{unwanted} for a user if she expresses distaste for it once seen within the recommendations.  
Users can express their distaste for certain recommendations by giving either a negative score or binary feedback (YouTube or Kuaishou \textit{``Don't recommend''} button \cite{mccrosky2021youtube}).
Similarly to \cite{wapo2020facebook}, we focus on the latter, and we consider a setting in which we can observe only binary negative feedback, as it happens for Kuaishou data (cf. \cref{sec:case-study-kuaishou}).
Let us denote with $H(U,I) \in \{0,1\}$ a binary random variable indicating if an item $I$ has been flagged ($H=1$) by the user $U$.
Given a set of recommendations $S_\lambda(U, k)$, we denote its overall \textit{risk} as the fraction of items the user has flagged once presented with it:
\begin{equation}
    R_H(S_\lambda(U, k)) = \frac{|\{ i \in S_\lambda(U, k) : H(U, I=i) = 1 \}|}{|S_\lambda(U, k)|}
    \label{eqn:set-based-harmfulness}
\end{equation}

\subsection{Conformal Risk Control}

By simply plugging the previous risk definition in \cref{eqn:standard-equation}, we have our optimization objective. 
However, how do we pick the threshold $\lambda$ to control the risk defined by \cref{eqn:set-based-harmfulness}?
We can find the optimal threshold via \textit{conformal risk control} \cite{angelopoulos2022conformal} by employing a held-out calibration set $\{(u,i,h))_j\}_{i=1}^n$, where for each user-item pairs we record also if it was flagged as \textit{unwanted} ($h=1$).
We first state the main theorem from \cite{angelopoulos2022conformal}, and then we explain its practical implications.
\begin{theorem}
    Assume that $R_H(S_\lambda(X, k))$ is non-increasing in $\lambda$, right-continuous,  $R_H(S_{\lambda_{max}}(U,k)) \leq \alpha$ and $\sup_\lambda R_H(S_\lambda(U,k)) \leq 1 < \infty$ almost surely.
    Consider $\hat{R}(\lambda) = \frac{1}{n}\sum_{j=1}^n R_H(S_\lambda(U=u_j, k))$, given a held-out calibration set.
    Given any desired risk upper bound $\alpha \in [0, 1]$, if we choose a threshold as:
    \begin{equation}
        \hat{\lambda} = \inf \left\{ \lambda : \frac{n}{n+1}\hat{R}(\lambda) + \frac{1}{n+1} \leq \alpha \right\}
    \end{equation}
    then, we have $\bbE[R(S_{\hat{\lambda}}(U,k))] \leq \alpha$.
    \label{theorem:conformal-risk-control}
\end{theorem}

In practice, as long as our risk function decreases as $\lambda$ increases, by choosing $\hat{\lambda}$ as in \cref{theorem:conformal-risk-control} we are guaranteed the risk will be \textit{provably} bound in expectation.
Practically speaking, if we consider as score the ranker output $s(U,I) = f(U,I)$, as we increase the threshold, fewer items will be contained within $T_\lambda(U)$.
Thus, the likelihood of picking harmful items within $S_\lambda(U, k)$ decreases monotonically since by removing items, we cannot increase $R_H(S_\lambda(U, k))$.  
Therefore, we can enjoy the risk-control guarantees provided by \cref{theorem:conformal-risk-control}\footnote{Our metric does not take into consideration the \textit{order} in which we present the item to the user, which might make $R_H(S_\lambda(U, k))$ non-monotone. In practice, we can monotonize any risk by taking $R_H(S_\lambda(U, k)) = \max_{\lambda' \geq \lambda} S_{\lambda'}(U, k)$ \cite{angelopoulos2022conformal}. 
}.
More importantly, the guarantees offered by \cref{theorem:conformal-risk-control} are model-agnostic and distribution-free. Therefore, they can be readily applied to \textit{any trained recommender}, given historical interaction data. 

\section{Risk Control via Item Replacement}

Recent applications of conformal prediction to recommender systems consider only the simple filtering approach as detailed in \cref{eqn:set-filtereing-T} \cite{angelopoulos2023recommendation,xu2024conformal,liang2024structured}.
The approach ensures the monotonicity requirements of \cref{theorem:conformal-risk-control} but at the price of losing control of the size of the resulting set $S_\lambda(U, k)$. 
First, notice that the sets are nested by $\lambda$ since $\lambda \geq \lambda' \Rightarrow T_{\lambda'}(U) \subseteq T_{\lambda}(U)$.    
If we use the ranker outputs as a scoring function, we will be able to influence the content of $S_\lambda(U,k)$ only when the threshold is greater than the score $f(U,I)$ of the $k$-th item in $\pi(I)$. 
Hence, in order to reduce the risk within $S_\lambda(U,k)$, we will have to choose a threshold $\lambda$ for which we will necessarily end up with fewer items than $k$. 
Indeed, we can state the following proposition:

\begin{proposition}
    Given a risk $\alpha \in [0,1]$ and $s(U,I)=f(U,I)$, there exist data distributions for which a thresholding rule, as in \cref{eqn:set-filtereing-T}, cannot provide the set with $H(S_\lambda(U,k)) \leq \alpha$ unless $|T_{\lambda}(U)|=0$.
    \label{prop:remove-based-harm-control-fail}
\end{proposition}

\begin{proof}
Consider the items $\calI = \{A,B,C,D\}$ where the ranker scores are $f(u, D) = 5$ and $f(u, i) = 1$ for $i \in \{A,B,C\}$. 
Thus, the ranker induces the following ordering $\pi(I) = \{D, A, B, C\}$ where we broke ties at random. 
Further, let us pick as a risk level $\alpha = 0.1$.
Then, it is easy to see how, for any $k \in [1,4]$, only with $\lambda > 5$ we have $H(S_\lambda(U,k) \leq 0.1$. 
Unfortunately, any threshold where $\lambda > 5$ implies $T_\lambda(I)=\varnothing$.
\end{proof}

Practically, it means that to provably control the associated risk of a recommendation, we might have to remove too many items from the recommendation pool.
Therefore, in our recommendation, we might not be able to return $k$ items to the user.
\cref{prop:remove-based-harm-control-fail} applies to \textit{any} score function. 
Consequently, this strategy is not optimal since, for example, in the case of a video streaming platform, we are essentially ``wasting'' items, thus potentially reducing engagement.

\begin{figure}[t]
    \centering
    \includegraphics[width=0.8\linewidth]{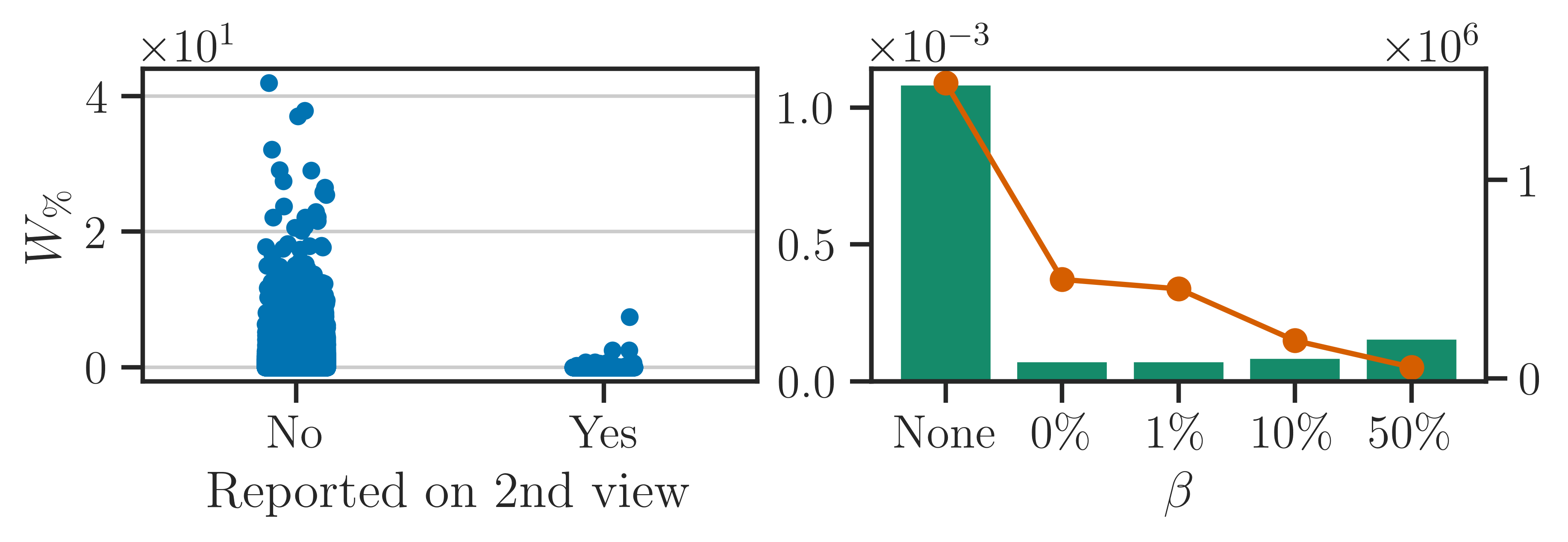}
    \caption{
    Validity of \cref{assumption:control_variable} in Kuaishou. (Left) Frequency of reported videos seen the second time based on their first watch time $W_{\%}$.
    (Right) Empirical $P(H = 1 \mid H_{1st}=0, W_\% > \beta)$ with varying $\beta$ ({\color{green_colorblind}green}) and total remaining repeated videos after thresholding over the watch-time ({\color{orange_colorblind} orange}). }
    \label{fig:assumption_1-checks}
\end{figure}

\subsection{A Property for Safe Alternatives}
\label{subsec:property-safe-alternatives}
Rather than simply removing items, we take an alternative approach, which consists of \textit{replacing} potentially unwanted items with safer ones, thus sidestepping the effects of \cref{prop:remove-based-harm-control-fail}.
However, how do we choose these \textit{safe items}?
Our approach exploits a simple fact: \textit{users consume again previously seen items} \cite{anderson2014dynamics,ren2019repeat,dai2024recode}.
Furthermore, as seen in Kuaishou data (\cref{tab:statistics-repeated-and-reported-videos}), users often dwell longer on already-consumed items, leading to greater engagement.
Inspired by the analysis in \cref{sec:case-study-kuaishou}, if we consider items the user has seen before, \textit{that were not reported as unwanted}, we might as well suggest them again in place of potentially unsafe items. 
As shown in \cref{tab:statistics-repeated-and-reported-videos}, there exist \textit{rare} items ($0.11\%$) which are reported only the second time a user watches them.
This poses a challenge since we might lose the monotonicity guarantees needed by \cref{theorem:conformal-risk-control}.
Therefore, we need a sensible way to select previously seen videos for which we are reasonably sure the user will not report anew.  
We formalize this desideratum in a simple property that can be tested by employing historical interaction patterns:
\begin{property}
    Denote with $H_{1st} \in \{0,1\}$ the event that the user has reported an item the first time they saw it.
    There exists a variable $C$ and a value $\beta \in \text{supp}(C)$ such that the probability of an item being flagged a second time is zero 
    $
      P(H = 1 \mid H_{1st}=0, C > \beta)  = 0
    $.
    \label{assumption:control_variable}
\end{property}
In the case of Kuaishou, we can pick as $C$ the watch time of videos.
Then, we can simply pick videos on which the users have spent time $W_{\%}(U, I)$ above a certain threshold. 
For example, \cref{fig:assumption_1-checks} shows the distribution of $P(H=1 \mid H_{1st} = 0, W_{\%})$. Indeed, the plot on the left shows that the videos that are reported in the second view have a much shorter watch time. 
Moreover, the right plot of \cref{fig:assumption_1-checks} shows that by filtering repeated videos with a thresholding $\beta$, we can reduce the likelihood of picking a harmful video almost to zero, globally. 
For $\beta = 50\%$, the probability slightly increases, since we have fewer repeated videos, and mostly because of a few outliers (\eg reported videos with a very high watch time). 
Practically, we provide an ablation study in \cref{sec:experiments-real-data} showing the effect of the choice of $\beta$ on risk control guarantees.

Thus, if \cref{assumption:control_variable} holds, we can simply replace items with a score below the risk control threshold, with items the user \textit{has previously} seen and not flagged as harmful, for which we have $C > \beta$. 
Let us denote the pool of \textit{safe} videos with the following:
\begin{equation}
    \Tsafe = \{ i' \in \calI_U : H(U,I=i') = 0 \\
    \wedge \; C(I=i') > \beta \}
\end{equation}
where $\calI_U \subseteq \calI$ indicates the items previously shown to the user $U$.
Here, we are assuming that the set $\calI_U$ is not empty, and more importantly, that there exist videos $i \in \calI_U$ such that $C(I=i) > \beta$.
A new user might not satisfy the conditions above, akin to \textit{cold start} issues \cite{park2009pairwise}.
In such a scenario, we might want to \textit{ask} the user to choose those videos for us.
%
%
Thus, the new pool of items will become:
\begin{equation}
  \begin{aligned}
    \Treplace_\lambda(U) = T_\lambda(U) \cup \Tsafe(U)
    \label{eqn:set_with_replaced_items}
\end{aligned}  
\end{equation}
It is straightforward to show that the new set $\Treplace_\lambda(U)$ satisfies the monotonicity requirements of conformal risk control, and it ensures we can always pick $k$ items to show to the user.

\begin{proposition} Given any $\lambda, \lambda' \in \Lambda$ such that $\lambda < \lambda'$.
Consider the set constructed with \cref{eqn:set_with_replaced_items}.
Then, if \cref{assumption:control_variable} holds, we have the following  
$
R_H\left(\Treplace_{\lambda'}(U)\right) \leq R_H\left(\Treplace_{\lambda}(U)\right)
$.
\label{prop:harm-decreasing-replacement}
\end{proposition}
We omit the proof since it readily follows from the fact that any repeated item has $H(U, I=i)=0$ if \cref{assumption:control_variable} holds. 

\begin{algorithm}[t]
\caption{Risk-controlling pipeline for a recommender system}\label{alg:cap}
\begin{algorithmic}[1]
\Require $u$, user, $\calI$, items, $\alpha$, risk level, $f: U\times I \rightarrow \bbR$, ranker, $s: U\times I \rightarrow \bbR^+$, score function, $k$ recommendation size
\State $\hat{\lambda} \gets \textsc{RiskControl}(\alpha)$ \Comment{\cref{theorem:conformal-risk-control}}
\State $T_{\hat{\lambda}}(U) \gets \{ y \in Y : s(U, I=i) \geq {\hat{\lambda}} \}$
\State $\Tsafe \gets \{ y' \in Y_X : H(U,I=i') = 0
    \wedge \; W_\%(U, i') > \beta \}$
\State $\Treplace_{\hat{\lambda}}(U) \gets T_{\hat{\lambda}}(U) \cup \Tsafe(U)$
\State \Return $\{ i_j \in \pi(\Treplace_{\hat{\lambda}}(U)) : j \leq k \}$
\end{algorithmic}
\label{alg:risk-control-replacement}
\end{algorithm}

\section{A Simple Post-Hoc Pipeline}

We consolidate the previous findings into a simple post-hoc algorithm applicable to \textit{any} pre-trained recommender system -- requiring \textit{no retraining} and the simple conditions outlined in \cref{theorem:conformal-risk-control}.
\cref{alg:risk-control-replacement} presents the pseudo-code of the procedure. Given a user-defined risk level $\alpha \in (0,1)$, the algorithm provably controls the fraction of unwanted content in the final recommendation list (\cf \cref{eqn:set-based-harmfulness}).
For a given user $u$ and risk level $\alpha$, we first compute the threshold $\hat{\lambda}$ using the \textsc{RiskControl} procedure from \cref{theorem:conformal-risk-control} (line 1). We then remove from the item pool all candidates with scores below $\hat{\lambda}$ (line 2) and identify potential replacements (line 3)—previously seen, non-flagged videos with watch-time exceeding a threshold $\beta$ (\cf \cref{assumption:control_variable}). These replacements are added to the remaining item pool (line 4), and the top-$k$ items are selected using the ranker (line 5).
Note that the final recommendation set may contain fewer than $k$ items if no suitable replacements are available (\ie if $\Tsafe$ is empty).
Indeed, \cref{alg:risk-control-replacement} reduces to classical risk control via removal only if we \textit{discard} line 3, and set $\Tsafe = \{\}$.
Lastly, the scoring function used for filtering may differ from the one used for reordering. For instance, risk control guarantees would still hold in a two-stage system where the score is the probability of \textit{not reporting} an item, and the ranker predicts watch time.

\textit{Computational complexity}.
The threshold $\hat{\lambda}$ can be precomputed for any risk level $\alpha$ and user $U$ with a complexity of $O(Q)$, where $Q$ is the size of the calibration set. \textsc{RiskControl} (line 1) can be run only once and can be cached for efficient retrieval at inference time.
Constructing the filtered set $T_{\hat{\lambda}}(U)$ and the replacement set $\Tsafe(U)$ (lines 2–3) requires a single pass over the item pool, with complexity $O(|I|)$. Re-ranking and selecting the top-$k$ items (line 5) has a complexity of $O(|I| \log |I|)$.
Thus, the overall complexity of \cref{alg:risk-control-replacement} is dominated by the sorting step, resulting in $O(|I| \log |I|)$.
\cref{alg:risk-control-replacement} is compatible with standard optimization techniques used in recommender systems (\eg efficient indexing, caching, or approximate nearest neighbours for lines 2–4). Naturally, the threshold $\hat{\lambda}$ should be periodically updated as new interaction data becomes available.

\begin{figure*}[t]
     \centering
     \begin{subfigure}[t]{0.95\textwidth}
          \centering
        \includegraphics[width=\linewidth]{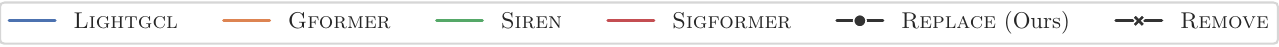}
     \end{subfigure}
     
     \begin{subfigure}[t]{0.40\textwidth}
          \centering
        \includegraphics[width=\linewidth]{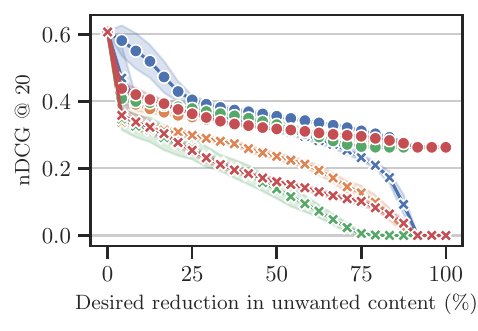}
        \caption{nDGC @ 20}
        \label{fig:synthetic-ndcg-experiments}
     \end{subfigure}
     \begin{subfigure}[t]{0.40\textwidth}
         \centering
          \includegraphics[width=\linewidth]{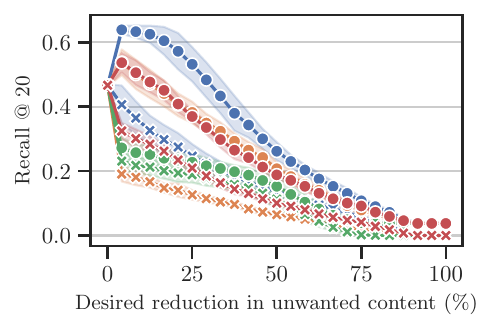}
          \caption{Recall @ 20}
          \label{fig:synthetic-recall-experiments}
     \end{subfigure}
     \hfill
     \begin{subfigure}[t]{0.40\textwidth}
         \centering
          \includegraphics[width=\linewidth]{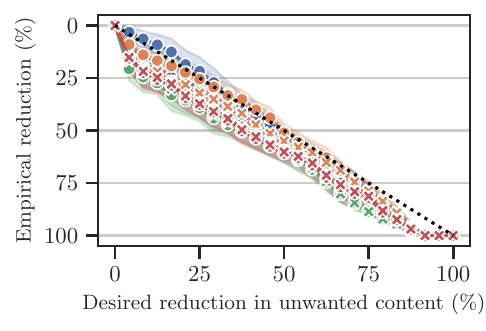}
          \caption{Reduction in $R_H(S_\lambda(U, k))$}
           \label{fig:synthetic-harmfulness-experiments}
     \end{subfigure}
     \begin{subfigure}[t]{0.40\textwidth}
         \centering
          \includegraphics[width=\linewidth]{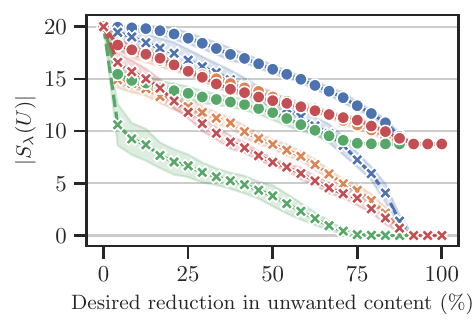}
          \caption{Size of $S_\lambda(U)$}
          \label{fig:experiments-real-size}
     \end{subfigure}
     \caption{Results over \textit{KuaiRand} ($k=20$ and $\beta = 0$).
     The results and standard deviation (shaded areas) are over 10 runs.
     }
\end{figure*}

\section{Experiments With Real Data}
\label{sec:experiments-real-data}

In this section, we evaluate the effects of the risk control strategies (\cref{eqn:set-filtereing-T,eqn:set_with_replaced_items}) on a simple two-stage recommender setup for a realistic video recommendation task.
In the following experiments, we aim to answer the following research questions:
\begin{itemize}
    \item [\textit{\textbf{RQ1}}] What is the impact of the conformal risk control on standard evaluation metrics (\eg \textit{Recall@k} and \textit{nDCG@k})?
    \item [\textit{\textbf{RQ2}}] What are the effects of employing different strategies for risk control (\cref{eqn:set-filtereing-T} and \cref{eqn:set_with_replaced_items})?
    \item [\textit{\textbf{RQ3}}] How many items do we have to replace from $S_{\hat{\lambda}}(U, k)$ to achieve the desired risk level? 
    \item [\textit{\textbf{RQ4}}] How does approximately satisfying \cref{assumption:control_variable} impact the desired unwanted content reduction?
    \item [\textit{\textbf{RQ5}}] If we divide the users according to their reporting habits in \textit{low-reporting} and \textit{high-reporting}, is \cref{alg:risk-control-replacement} reducing the fraction of unwanted content also for \textit{high-reporting} users?
\end{itemize}


\textit{\textbf{Dataset}.}
We use the KuaiRand dataset \cite{kuairand2022}, which contains both positive and negative user feedback, along with information about interactions with previously seen videos.
As discussed in \cref{sec:unwanted-disliked-harmful}, to the best of our knowledge, it is the only publicly available dataset that includes realistic feedback signals and a target metric such as video watch time.
For our experiments, we focus on videos viewed at least twice by at least one user, resulting in a dataset that includes both repeated and single interactions.
This leaves us with 5657 unique users, 117695 unique items, and more than 3 million interactions.
Then, we first partition the full interaction set into two datasets based on whether the interaction is repeated or not.
We adopt a 10-core setting, randomly splitting the subset with only single interactions into training (70\%), validation (15\%), and test (15\%) sets.
The validation set also serves as a calibration set for determining the threshold $\lambda$ as defined in \cref{theorem:conformal-risk-control}.
To ensure evaluation fidelity, we exclude all test set videos from the repeated interaction set.
Thus, the repeated interaction set only considers videos previously seen by users — \ie those appearing in either training or validation.

\textit{\textbf{Metrics}.}
We compare the effect of the risk control and various strategies on two widely used metrics: \textit{Recall@k} and \textit{nDCG@k}. In our experiments, we set $k=20$ as usually done in the literature~\cite{chen2024sigformer}. 
We also measure the empirical \textit{risk}, \ie the fraction of unwanted content in the recommendations, as defined in \cref{eqn:set-based-harmfulness}. 


\textit{\textbf{Models}.}
Similarly to previous studies \cite{wang2023uqfairness}, we evaluate the effectiveness of risk control using a two-stage recommender architecture: (a) the first stage filters items using the conformal risk control procedure described in \cref{eqn:set_with_replaced_items}; (b) the second stage re-ranks the filtered items using a model trained to predict watch-time $W(U, Y)_\%$.
We consider two risk-control strategies: \textsc{Remove} (\cref{alg:risk-control-replacement} without line 3 and $\Tsafe = \{\}$) and \textsc{Replace} (\cref{alg:risk-control-replacement}).
\textsc{Remove} is the classical risk-control strategy for recommender systems \cite{angelopoulos2023recommendation, xu2024conformal}.
For simplicity, under the \textsc{Replace} strategy, repeated items are not re-scored; instead, their original watch time is reused.
To evaluate the impact of different scoring functions $s(U, Y)$ on the effectiveness of risk control, we experiment with several state-of-the-art methods.
These include both sign-aware and unsigned graph-based models, depending on whether they incorporate negative user feedback (\eg reported videos):
\textbf{LightGCL} \cite{cai2023lightgcl} and \textbf{GFormer} \cite{li2023gformer}, state-of-the-art graph-based methods leveraging contrastive learning and the transformer architecture, respectively.
\textbf{SiReN} \cite{seo2022siren}, the classic sign-aware method for recommendation learning two embeddings from positive and negative interaction graphs, and \textbf{SIGFormer} \cite{chen2024sigformer}, a state-of-the-art sign-aware method employing the transformer architecture.
For the re-ranking stage, we adopt a simple Neural Collaborative Filtering (NCF) model \cite{he2017neural}, as our focus is not on maximizing performance on the Kuaishou dataset, but rather on isolating the effects of risk control.
Model parameters were selected according to the original papers and implemented using the authors’ source code.
Our full implementation is openly available\footnote{\url{https://github.com/geektoni/mitigating-harm-recsys}} to support reproducibility and further research.

\begin{figure}[t]
     \centering
     \includegraphics[width=0.65\linewidth]{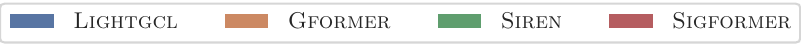}
     \includegraphics[width=0.55\linewidth]{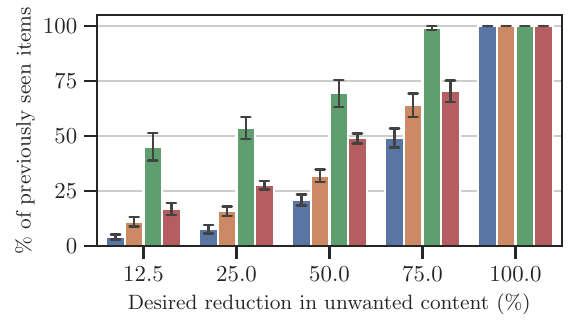}
     \caption{Fraction of previously seen items in the top-20 recommendations for Kuaishou ($k=20$ and $\beta = 0\%$).
     }
     \label{fig:repeated-items-in-recommendation}
\end{figure}

\textit{\textbf{(RQ1 \& RQ2) Replacing unwanted items ensures risk control and mitigates performance degradation}.}
We start by evaluating the impact of risk control on standard ranking metrics, \ie nDCG@20 and Recall@20.
As shown in \cref{fig:synthetic-harmfulness-experiments}, the level of harmfulness is tightly controlled across all strategies and ranking models.
For any target reduction in unwanted content ($\alpha$ in \cref{alg:risk-control-replacement}), the empirical reduction on the test set is equal to or better (\ie all points lie on or below the diagonal) than the desired target, confirming that \cref{alg:risk-control-replacement} reliably enforces risk control.
In general, the \textsc{Remove} strategy tends to exceed the target reduction, removing more items than strictly necessary.
This proves that \cref{alg:risk-control-replacement} provides a provable link between user actions (\eg flagging content) and the resulting changes in recommendations.
Indeed, by specifying a desired reduction level, users can provably influence their experience with guarantees on the expected reduction in unwanted content.
\cref{fig:synthetic-ndcg-experiments,fig:synthetic-recall-experiments} show how risk control affects nDCG@20 and Recall@20.
Removing items sharply degrades nDCG as the desired reduction level ($\alpha$ in \cref{eqn:standard-equation}) increases, due to the shrinking pool of candidate items (\cref{eqn:set-filtereing-T}).
In contrast, the \textsc{Replace} strategy retains a higher nDCG, though some degradation still occurs as previously seen items are reintroduced into the ranking.
This performance drop is partly explained by our experimental setup: we do not re-score repeated items but instead reuse their original watch time, which may not reflect user interest upon second exposure.
Recall decreases under both strategies, as removing or replacing items may discard potentially relevant, unseen content.
Finally, \cref{fig:experiments-real-size} illustrates the consequences of \cref{prop:remove-based-harm-control-fail}. When full removal of unwanted content is required ($\alpha = 0$), the sole option is to eliminate or replace all items in the recommendation set.
Moreover, for the \textsc{Replace} strategy, we might recommend less than $k$ items at test time, since we have finite replacement options. 

\begin{figure*}[t]
     \centering
     \begin{subfigure}[t]{0.40\linewidth}
          \centering
        \includegraphics[width=\linewidth]{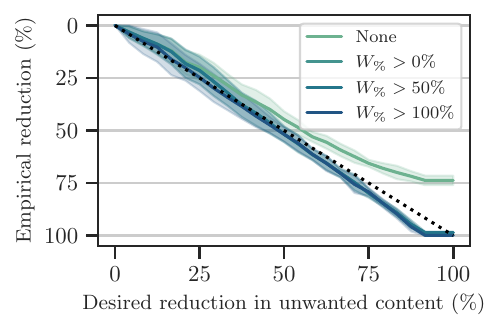}
        \caption{Reduction in $R_H(S_\lambda(U, k))$}
        \label{fig:ablation-beta-harmfulness}
     \end{subfigure}
     \begin{subfigure}[t]{0.40\linewidth}
         \centering
          \includegraphics[width=\linewidth]{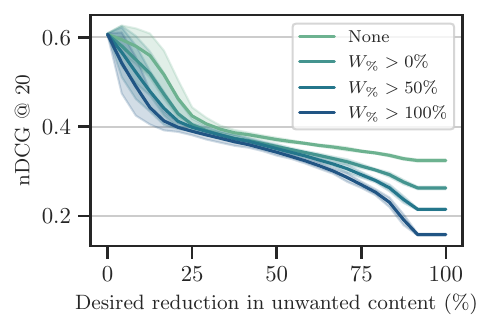}
          \caption{nDCG @ 20}
          \label{fig:ablation-beta-ndcg}
     \end{subfigure}
     \begin{subfigure}[t]{0.40\linewidth}
         \centering
          \includegraphics[width=\linewidth]{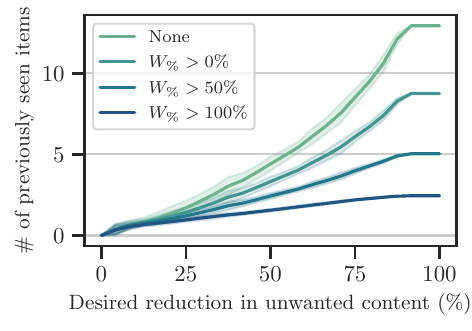}
          \caption{\# of previously seen items}
          \label{fig:ablation-beta-replacement}
     \end{subfigure}
     \begin{subfigure}[t]{0.40\linewidth}
         \centering
          \includegraphics[width=\linewidth]{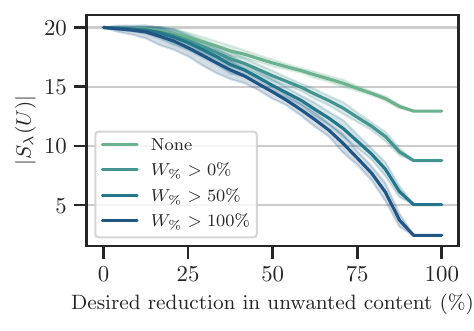}
          \caption{Size of $S_\lambda(U)$}
          \label{fig:ablation-beta-size}
     \end{subfigure}
     \caption{Ablation analysis on $\beta$ for \textsc{LightGCL} ($k=20$). 
     We consider only previously seen videos with a watch-time greater than $W_\% > \{0\%, 50\%, 100\%\}$, or without filtering (\textsc{None}). 
     The results and standard deviation (shaded areas) are computed over 10 runs.
     }
     \label{fig:ablation-experiments-beta}
\end{figure*}

\textit{\textbf{(RQ3) The choice of score function impacts the number of previously seen items in the recommendation set}.}
We examine the average fraction of previously seen items recommended to users.
While the risk control guarantees of \cref{theorem:conformal-risk-control} hold independently of the underlying ranking function, in practice, we favor models that minimize item replacements at an equal level of risk control.
\Cref{fig:repeated-items-in-recommendation} shows that sign-aware models (\textsc{SiReN} and \textsc{SIGformer}) lead to more replaced items compared to their unsigned counterparts (\textsc{LightGCL} and \textsc{GFormer}) for equivalent reductions in unwanted content.
This is somewhat unexpected, as sign-aware models incorporate both positive and negative feedback, whereas the unsigned models do not.
However, due to the low ratio of negative to positive feedback in the Kuaishou dataset, the additional information appears insufficient to improve ranking.
For example, \textsc{SiReN} may struggle to learn meaningful embeddings from the highly sparse negative interaction graphs.
As a result, sign-aware models exhibit higher uncertainty and require replacing or removing more items to meet risk control guarantees. In contrast, \textsc{LightGCL} and \textsc{GFormer} produce larger recommendation sets (\cf \cref{fig:experiments-real-size}) while minimizing item modifications.
In the next experiments, we focus on \textsc{LightGCL} since it emerged as the best scoring function.

\textit{\textbf{(RQ4) There exists a trade-off between ensuring the safety of replacement items and the recommendation set size}.}
Further, we investigate the impact of relaxing \cref{assumption:control_variable}, which is required to satisfy \cref{prop:harm-decreasing-replacement}.
Specifically, we conduct an \textit{ablation study} to assess the effect of the approximate satisfaction of this assumption.
\cref{fig:ablation-experiments-beta} shows the effect of increasing the filtering threshold $\beta$ to select safer replacement items.
As shown in \cref{fig:ablation-beta-harmfulness}, when no filtering is applied (\textsc{None}), risk control fails — indeed, some previously seen items, when recommended again, are revealed to be unwanted and reported by users (as shown in \cref{tab:statistics-repeated-and-reported-videos}).
This behavior reflects a form of distribution shift between calibration and test phases \cite{angelopoulos2022conformal}.
Conversely, filtering items based on their initial watch-time ($W_\% > {0\%, 50\%, 100\%}$) restores risk control, as seen by results falling on or below the diagonal in \cref{fig:ablation-beta-harmfulness}.
However, stricter filtering reduces the pool of candidate items, approaching the behavior of the \textsc{Remove} strategy, as indicated in \cref{fig:ablation-beta-replacement,fig:ablation-beta-size}.
This reveals a trade-off between ensuring risk control and maintaining a sufficiently large recommendation set under finite replacement options. 

\textit{\textbf{(RQ5) Risk control is biased toward high-reporting users}.}
Lastly, we analyze the impact of the risk control procedure across user groups.
We can categorize users as being \textit{low-reporting}, who report a smaller fraction of videos ($H_{X} < 0.1$, cf. \cref{tab:statistics-reported-videos}) versus \textit{high-reporting}, who report a higher fraction of videos ($H_{X} \geq 0.1$).
\cref{fig:ablation-experiments-users-group} presents the empirical reduction of unwanted content for both groups under different strategies.
Notably, \textit{low-reporting} users exhibit a greater reduction at lower $\alpha$ values, while other performance metrics (\eg nDCG in \cref{fig:ablation-user-ndcg}) remain stable across groups and strategies.
This suggests that \textit{high-reporting} users disproportionately influence the calibration of the risk threshold $\lambda$, potentially leading to overly conservative behavior for everybody.
In practice, equivalent risk guarantees for \textit{low-reporting} users could be achieved by modifying or removing \textit{fewer items}.

\begin{figure}[t]
     \centering
     \begin{subfigure}[t]{0.40\linewidth}
          \centering
        \includegraphics[width=\linewidth]{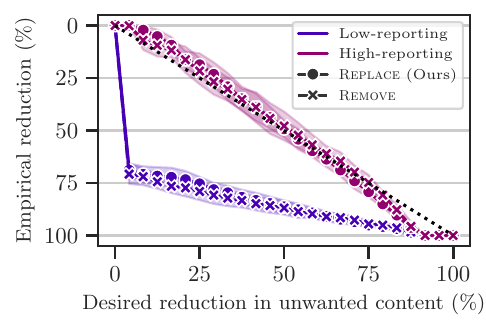}
        \caption{Reduction in $R_H(S_\lambda(U, k))$}
        \label{fig:ablation-user-harm}
     \end{subfigure}
     \begin{subfigure}[t]{0.40\linewidth}
          \centering
        \includegraphics[width=\linewidth]{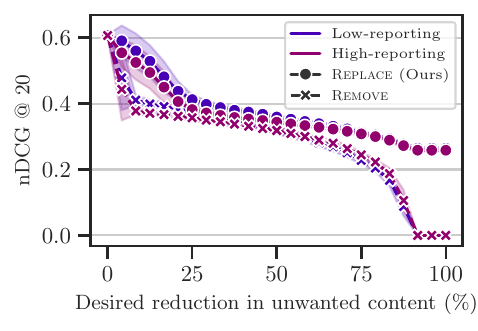}
        \caption{nDCG @ 20}
        \label{fig:ablation-user-ndcg}
     \end{subfigure}
     \caption{Ablation analysis on \textit{low}- or \textit{high}-reporting users (\textsc{LightGCL}, $k=20$ and $\beta = 0\%$). 
     We report the results over 10 runs, where the shaded area indicates the standard deviation.
     }
     \label{fig:ablation-experiments-users-group}
\end{figure}

\section{Discussion and Limitations}
\label{sec:discussion-limitations}

In this section, we discuss some assumptions and limitations of our work, which open up interesting
avenues for future work.

\textit{Methodology.}
The risk control procedure (\cref{alg:risk-control-replacement}) computes a global threshold that minimizes expected risk across users.
Although users can select and reliably achieve a target level of unwanted content reduction, a single global threshold affects user groups differently (\eg \textit{low-} vs. \textit{high-reporting} behavior).
In practice, it may be preferable to personalize the threshold, akin to \textit{group-balanced} conformal prediction \cite{angelopoulos2023gentleintro}, by calibrating $\hat{\lambda}$ using data from a single user, provided sufficient interactions are available.
We also assume that user preferences over unwanted content remain stable over time; future work could relax this assumption by modeling dynamic user profiles \cite{chee2024harm}.
Moreover, the current risk function (\cref{eqn:set-based-harmfulness}) does not account for item position in the ranking and treats each item as contributing equally to risk.
Extensions could incorporate user models such as \textit{discrete choice models} \cite{krause2025exposure-bias-discrete}, enabling more realistic risk functions.
Finally, the guarantees in \cref{theorem:conformal-risk-control} hold \textit{in expectation}; further extensions could include probabilistic guarantees \cite{bates2021distribution, angelopoulos2023recommendation} or support for non-monotonic risk \cite{angelopoulos2021learn}.

\textit{Evaluation}.
Our analysis and experiments rely on a single dataset, as discussed in \cref{sec:unwanted-disliked-harmful}, which currently offers the only realistic basis for studying this phenomenon.
This reliance may limit the external validity of our findings.
Further evaluation, particularly of assumptions such as \cref{assumption:control_variable}, will require additional datasets. However, acquiring such data remains a significant challenge. 
Regulatory frameworks may help address this gap. For example, Article 40 of the European Digital Services Act\footnote{\url{https://eur-lex.europa.eu/eli/reg/2022/2065/oj}}\footnote{\url{https://data-access.dsa.ec.europa.eu/home}} permits vetted researchers to access data from very large online platforms (\eg Facebook, Twitter) for scientific purposes, potentially enabling validation of \cref{alg:risk-control-replacement}.
Nonetheless, privacy and ethical considerations remain crucial, requiring appropriate data handling and anonymization.

\textit{User experience.}
Lastly, replacing certain items in the recommendation pool with previously seen content may lead to unbalanced recommendations, repeatedly exposing users to the same content.
This could have unintended negative effects, such as increased polarization \cite{cinus2022effect,tommasel2021want}.
{\color{black}
While our chosen replacement strategy (see Section~\ref{subsec:property-safe-alternatives}) is relatively conservative, future work could explore alternatives based on curated item lists from the community or user-provided signals, such as liked videos or watch-later queues.
Additionally, recent research on repeated consumption~\cite{wang2019modeling,ren2019repeat,dai2024recode} offers targeted methods that may also address the balance between risk control and recommendation fairness~\cite{wang2023uqfairness}.
Interestingly, alternative approaches could adopt a more systemic perspective.
For example, a video-sharing platform might aggregate user feedback to identify content to downrank, allowing individuals to share “unwanted” content signals within trusted groups, such as friends or communities. 
Users with limited feedback activity could opt into these shared signals, effectively leveraging collective judgments to inform their own recommendations.
This could help address cold-start issues~\cite{park2009pairwise} in the context of safe repeated items.
Finally, user studies could yield more robust empirical evidence on how repetition and replacement strategies affect satisfaction, guiding further development in this area.
}

\section{Conclusions}

This work introduces a simple but effective user-centric procedure (\cf \cref{alg:risk-control-replacement}) grounded in conformal risk control to mitigate the fraction of unwanted content suggested by any recommender systems.
Unlike prior approaches that offer limited guarantees, our method proactively bounds the presence of unwanted content in recommendations by ensuring a connection between user feedback and change in the recommender behavior.
The procedure is model-agnostic, easy to integrate with existing systems, and backed by distribution-free theoretical guarantees.
We provide a detailed analysis of real-world interaction patterns (\eg reporting behavior and watch time) to inform the design of our approach.
We further introduce a practical heuristic for selecting repeated content that balances safety and relevance. 
Experiments on a real-world dataset demonstrate that our framework can controllably reduce exposure to unwanted recommendations by providing a trade-off between the recommendation set size and utility.
Finally, these findings point toward developing safer and user-centric recommender systems.

\section*{Acknowledgments}

The work was partially supported by the following projects: Horizon Europe Programme, grants \#10112\-0237-ELIAS and \#101120763-TANGO.
Funded by the European Union. Views and opinions expressed are however those of the author(s) only and do not necessarily reflect those of the European Union or the European Health and Digital Executive Agency (HaDEA). Neither the European Union nor the granting authority can be held responsible for them.

\bibliography{references}
\bibliographystyle{plainnat}

\end{document}